\documentclass[11pt]{article}
\usepackage{amssymb}
\makeatletter\@addtoreset {equation}{section}\makeatother

\newtheorem{theorem}{Theorem}
\newtheorem{lemma}{Lemma}
\newtheorem{remark}{Remark}

\newtheorem{corollary}{Corollary}

\newtheorem{proposition}{Proposition}
\def\R{\mathbb{R}}

\newenvironment{proof}{
    \noindent {\it Proof.}}{\hfill$\Box$
}

\usepackage[dvips]{epsfig}
\usepackage{graphicx}

\begin{document}

\title{\bf Rigorous justification of the short-pulse equation}

\author{Dmitry Pelinovsky$^1$ and Guido Schneider$^2$ \\
{\small $^1$ Department of Mathematics and Statistics, McMaster University} \\ {\small
1280 Main Street West, Hamilton, Ontario, Canada, L8S 4K1} \\
{\small $^2$ Institut f\"{u}r Analysis, Dynamik und Modellierung,
Universit\"{a}t Stuttgart} \\ {\small Pfaffenwaldring 57, D-70569
Stuttgart, Germany } }

\date{\today}
\maketitle

\begin{abstract}
We prove that the short-pulse equation, which is derived from a
quasilinear Klein--Gordon equation with formal asymptotic methods, can be rigorously
justified. The justification procedure applies to small-norm
solutions of the short-pulse equation. Although the small-norm solutions
exist for infinite times and include modulated pulses and their elastic interactions,
the error bound for arbitrary initial data can only be controlled over finite time intervals.
\end{abstract}

\section{Introduction}

Short pulses play an important role in nonlinear optics \cite{Vladimirov,Schafer},
nonlinear meta-materials \cite{Kevrekidis}, and mode-locked lasers \cite{Kutz}. The classical envelope equations
such as the nonlinear Schr\"{o}dinger equation are no longer valid as the pulse width is
only few carrier wavelengths, instead of thousands of these. Short-pulse approximations
have been derived in this context by using geometric optics \cite{Y1,Y2},
diffractive nonlinear optics \cite{AR02,AR03},
nonlocal envelope equations with full dispersion \cite{BL02,CGL05},  and
a regularized nonlinear Schr\"{o}dinger equation \cite{CL09}. These models have been
rigorously justified similarly to the justification procedure of the classical
nonlinear Schr\"{o}dinger equation \cite{Kal,Sch}. Under the term ``rigorous justification",
we understand that the error between solutions of the original and approximated equations
is controlled in some norm over sufficiently long time intervals.

A different model for short pulses with few cycles on the pulse width
was derived by Sch\"afer \& Wayne \cite{SW04}. We term this model as
{\em the short-pulse equation} and write it in the form,
\begin{equation}
A_{\xi \tau} = A + (A^{3})_{\xi \xi},
\label{short-pulse}
\end{equation}
where $\tau \in \R_+$ is the evolution time, $\xi \in \R$ is the spatial coordinate, and
$A(\tau,\xi) \in \mathbb{R}$ is the amplitude function. We emphasize that
this short-pulse equation is different from all short-pulse approximations
used earlier. It is dispersive compared to the geometric optics \cite{Y1,Y2}, it is
quasilinear compared to the diffractive nonlinear optics \cite{AR02,AR03},
and it is not an envelope equation for nearly harmonic linear waves compared
to the nonlocal and regularized nonlinear Schr\"{o}dinger equations \cite{BL02,CL09,CGL05}.
Although the short-pulse equation (\ref{short-pulse}) is a one-dimensional model, it can be generalized
to the two- and three-dimensional geometries, for the price of losing all the nice properties
of this short-pulse equation listed next.

The short-pulse equation (\ref{short-pulse}) represents the class of
nonlinear wave equations with low-frequency dispersion, which reduce
in the dispersionless limit to the inviscid Burgers equation. Local
well-posedness of the short-pulse equation was established in $H^s(\R)$
for $s > \frac{3}{2}$ \cite{SW04,SSK10}. Using a hierarchy of
conserved quantities of the short-pulse equation \cite{Br05},
solutions with small $H^2$ norm were extended globally for infinite time \cite{PS10}.
On the other hand, solutions with large $H^2$ norm were proved to blow up in a finite time
\cite{LPS09}. The blow-up behavior
resembles wave breaking when the amplitude $A$ remains bounded but the slope
steepens up, similar to the self-steeping behavior of the inviscid Burger equation.

Sakovich \& Sakovich found that the short-pulse equation (\ref{short-pulse}) is integrable
by means of the inverse scattering transform \cite{SS05}. By a coordinate
transformation, this equation is reduced to
the sine--Gordon equation in characteristic coordinates, which admits
exact modulated pulse (breather) solutions \cite{SS06}.
Multi-pulse solutions exhibiting elastic scattering as well as periodic wave solutions of the short-pulse equation
were later found by Matsuno \cite{Matsuno1,Matsuno}.

It is the purpose of this article to justify the applicability of the short-pulse
equation (\ref{short-pulse}) to dynamics of pulses in the framework of the quasilinear Klein--Gordon
equation,
\begin{equation}
\label{Maxwell}
u_{tt} - u_{xx} + u + (u^3)_{xx} = 0,
\end{equation}
where $t \in \R_+$, $x \in \R$, and $u(t,x) \in \R$. Compared to the full
system of Maxwell equations in electromagnetic
theory \cite{SW04}, the quasilinear Klein--Gordon equation (\ref{Maxwell}) is only the toy model.
It has been used before to construct the breather solutions on a finite spatial
scale by means of the spatial dynamics methods \cite{GS05}.

Regarding justifications of the short-pulse equation, Chung {\em et al.} \cite{CJSW05}
developed the justification analysis for the linear version
of the short-pulse equation by working with oscillatory integrals and roots
of the dispersion relations in a more complicated system of Maxwell equations.
They also illustrated numerically that the nonlinear version of the short-pulse
equation, derived heuristically with a formal renormalization procedure, yields a very
good approximation of the modulated pulse solutions in the limit of few cycles
on the pulse width. However, the problem of justification of the nonlinear
short-pulse equation (\ref{short-pulse}) remained opened up to the date.
We emphasize that the justification analysis, albeit similar to other short-pulse
approximations, does not follow immediately from earlier literature on the subject.

To develop the nonlinear justification analysis of the short-pulse equation (\ref{short-pulse}), we use
the local existence theory and apriori energy estimates. Dealing with the energy estimates,
it is difficult to control the solutions of the quasilinear Klein--Gordon equation (\ref{Maxwell})
in Sobolev spaces $H^s(\R)$ with higher index $s > 0$ using the scaled variables of
the short-pulse equation (\ref{short-pulse}). These norms diverge as $\epsilon \to 0$,
the higher is the index, the faster is the divergence. To avoid this difficulty,
we shall implement the coordinate transformation from the beginning and work
with the error term in the scaled variables. Specifically, we use the transformation of
variables,
\begin{equation}
\label{new-variables}
u(t,x) = 2 \epsilon U(\tau,\xi), \quad \tau = \epsilon t, \quad \xi = \frac{x - t}{2 \epsilon},
\end{equation}
and rewrite the quasilinear Klein--Gordon equation (\ref{Maxwell}) in the equivalent form,
\begin{equation}
\label{Maxwell-short}
U_{\tau \xi} = U + (U^3)_{\xi \xi} + \epsilon^2 U_{\tau \tau}.
\end{equation}
The short-pulse equation (\ref{short-pulse}) appears from equation
(\ref{Maxwell-short}) by neglecting the last term $\epsilon^2 U_{\tau \tau}$.
It is the starting point of our analysis. The following theorem
presents the main result.

\begin{theorem}
For all  $s > \frac{7}{2}$ and $T > 0$, there exists a $\delta_0 > 0$ such
that for all $\delta \in (0,\delta_0)$, there exist $ \epsilon_0 > 0 $ and $ C_0 > 0 $
such that for all $ \epsilon \in (0,\epsilon_0) $, the following holds.
Let $A \in C([0,T],H^s(\R))$
be a local solution of the short-pulse equation (\ref{short-pulse}) such
that
\begin{eqnarray}
\nonumber
& \phantom{t} & \sup_{\tau \in [0,T]} \| A(\tau,\cdot) \|_{H^s} + \sup_{\tau
\in [0,T]} \| A_{\tau}(\tau,\cdot) \|_{H^{s-1}}
+ \sup_{\tau \in [0,T]} \| A_{\tau \tau}(\tau,\cdot) \|_{H^{s-2}} \\
\label{condition-solution}
&  \phantom{t} & \phantom{texttexttexttexttexttexttexttexttext}
+ \sup_{\tau \in [0,T]} \| A_{\tau \tau \tau}(\tau,\cdot) \|_{H^{s-3}} \leq
\delta
\end{eqnarray}
and let $U_0 \in H^3(\R)$ and $V_0 \in H^2(\R)$ be such that
\begin{equation}
\label{error-initial}
\left\| U_0 - A(0,\cdot) \right\|_{H^2} +
\left\| V_0 - A_{\tau}(0,\cdot) \right\|_{H^1} \leq \epsilon.
\end{equation}
Then there exists a unique solution
\begin{eqnarray*}
U \in C([0,T],H^2(\R)) \cap C^1([0,T],H^1(\R)) \cap C^2([0,T],L^2(\R))
\end{eqnarray*}
of the quasilinear Klein--Gordon equation (\ref{Maxwell-short}) subject to the initial data
$U(0,\cdot) = U_0$, $U_{\tau}(0,\cdot) = V_0$ satisfying
\begin{equation}
\label{error-final}
\sup_{\tau \in [0,T]} \left\| U(\tau,\cdot) - A(\tau,\cdot) \right\|_{H^2}
\leq C_0 \epsilon.
\end{equation}
\label{theorem-main}
\end{theorem}

\begin{remark}
Condition (\ref{condition-solution}) can be satisfied from constraints on the initial data
of the short-pulse equation (\ref{short-pulse}), see Corollaries \ref{corollary-first-antider},
\ref{corollary-second-antider}, and \ref{corollary-third-antider}. Loosely speaking,
these constraints are satisfied when the first
three anti-derivatives of $A(0,\xi)$ in $\xi$ are square integrable.
Note that these constraints are only imposed on the initial condition $A(0,\xi)$
of the short-pulse equation (\ref{short-pulse}), the initial condition of the
quasilinear Klein--Gordon equation (\ref{Maxwell-short}) is arbitrary within
the proximity bound (\ref{error-initial}).
\end{remark}

\begin{remark}
For semilinear hyperbolic systems, Alterman \& Rauch \cite{AR03}
derived and justified the semilinear short-pulse equation
without requiring any constraints on  the initial data. They introduced
an $\epsilon$-dependent cut-off function in Fourier space removing the
Fourier modes close to the zero wave numbers. However, our original
system and our short-pulse equaiton are quasilinear and most of the subsequent
technical difficulties in our analysis are due to this fact. For the same reason,
it is not clear if the constraints on the initial data of the short-pulse
equation (\ref{short-pulse}) can be removed by using a similar cut-off function.
The answer to this question will be a subject of further studies.
\end{remark}

\begin{remark}
In terms of the variables of the original equation (\ref{Maxwell}), we can rewrite
bounds (\ref{error-initial}) and (\ref{error-final}) in the equivalent form.
If we assume that, for sufficiently small $\epsilon > 0$,
 there is $C > 0$ such that the initial data satisfy
\begin{equation}
\label{error-initial-unscaled}
\left\| u(0,\cdot) - 2 \epsilon A\left(0,\frac{\cdot}{2 \epsilon}\right) \right\|_{H^2} \leq C \epsilon^{1/2},\quad
\left\| u_t(0,\cdot) + A_{\xi}\left(0,\frac{\cdot}{2 \epsilon}\right) \right\|_{H^1} \leq C \epsilon^{1/2},
\end{equation}
then there exists a $T$-dependent but $\epsilon$-independent positive constant $C_0$ such that
the solution of the quasilinear Klein--Gordon equation (\ref{Maxwell}) satisfies
\begin{equation}
\label{error-final-unscaled}
\sup_{t \in [0,T/\epsilon]} \left\| u(t,\cdot) - 2 \epsilon A\left(\epsilon t,\frac{\cdot - t}{2 \epsilon}\right) \right\|_{H^2} \leq C_0 \epsilon^{1/2}.
\end{equation}
If $A_0 \in H^2(\R)$ and $\epsilon \to 0$, then
\begin{equation}
\left\| \epsilon A_0 \left(\frac{\cdot}{2 \epsilon}\right) \right\|_{H^2} = \mathcal{O}(\epsilon^{-1/2}), \quad
\left\| A_0'\left(\frac{\cdot}{2 \epsilon}\right) \right\|_{H^1} = \mathcal{O}(\epsilon^{-1/2}).\label{leading}
\end{equation}
Bounds (\ref{error-initial-unscaled}), (\ref{error-final-unscaled}), and (\ref{leading}) show that the error
terms between solutions of the short-pulse equation and the quasilinear Klein--Gordon equation
in the original variables are also $\mathcal{O}(\epsilon)$ smaller than the leading order terms.
\end{remark}

{\bf Organization of the paper:} Section 2 presents results about local and global solutions
of the short-pulse equation (\ref{short-pulse}). In particular, we identify the constraints
on the initial conditions $A(0,\xi)$, which verify the validity of the assumption (\ref{condition-solution})
in Theorem \ref{theorem-main}. Section 3 deals with local solutions
of the quasilinear Klein--Gordon equations (\ref{Maxwell}) and (\ref{Maxwell-short}).
We derive the continuation criterion for local solutions of the quasilinear equations,
which is useful to extend local solutions to the times $t = \mathcal{O}(1/\epsilon)$
or $T = \mathcal{O}(1)$ required for the justification result in Theorem \ref{theorem-main}.
Section 4 reports apriori energy estimates for the error term between solutions of the
short-pulse and quasilinear Klein--Gordon equations. The justification result
of Theorem \ref{theorem-main} is proven in Section 5 by using a continuation
argument together with the energy estimates.

{\bf Notations:} $H^s(\R)$ for $s \geq 0$ denotes the Hilbert--Sobolev space equipped with
the norm
$$
\| f \|_{H^s} = \left( \int_{\R} (1 + k^2)^{s} |\hat{f}(k)|^2 dk \right)^{1/2},
$$
where $\hat{f}$ is the Fourier transform of $f$.
We shall intersect these spaces with $\dot{H}^{-m}$, $m \in \mathbb{N}$, equipped with the norm
$$
\| f \|_{\dot{H}^{-m}} = \left( \int_{\R} k^{-2m} |\hat{f}(k)|^2 dk \right)^{1/2}.
$$
If $f \in \dot{H}^{-m}(\R)$, then the $m$-th order anti-derivative of $f$ is square integrable.

We define the anti-derivative of $f \in L^2(\R) \cap \dot{H}^{-1}(\R)$ by
$$
\partial_{\xi}^{-1} f := \int_{-\infty}^{\xi} f(\xi') d \xi'.
$$
Under the condition $f \in L^2(\R) \cap \dot{H}^{-1}(\R)$, the anti-derivative of $f$ is not only square integrable, but also continuous
and decaying to zero as $|\xi| \to \infty$ thanks to Sobolev embedding. In particular,
$f$ is the mean-zero function satisfying the constraint $\int_{-\infty}^{\infty} f(\xi) d \xi = 0$.

Constant $C$ stands for a generic $\epsilon$-independent positive constant, which may change from
one line to another line and from one term to another term in the same inequality.

{\bf Acknowledgments:} This project was initiated during the workshop
on the short-pulse equations organized at the Fields Institute (May, 2011).
D. Pelinovsky is partially supported by
the Alexander von Humboldt Foundation. G. Schneider is partially supported by the Deutsche
Forschungsgemeinschaft (DFG) grant SCHN 520/8-1.

\section{Local solutions of the short-pulse equation}

We start with the local well-posedness of the short-pulse equation (\ref{short-pulse}).
An improved local existence result is obtained by Stefanov {\em et al.} \cite[Theorem 1]{SSK10}.
The following statement will be used in the estimates for the error terms
generated by the local solutions of the short-pulse equation.
In particular, we will specify constraints on the initial data $A(0,\cdot)$
of the short-pulse equation, which would guarantee the existence of a local solution satisfying
the bound (\ref{condition-solution}) assumed in Theorem \ref{theorem-main}.

\begin{proposition}\cite{SSK10}
\label{lemma-existence-short}
Fix $s > \frac{3}{2}$. For any $A_0 \in H^s(\R)$, there exists a time $\tau_0 = \tau_0(\| A_0 \|_{H^s}) > 0$ and
a unique strong solution of the short-pulse equation (\ref{short-pulse}) such that
\begin{equation}
\label{local-solution}
A \in C([0,\tau_0],H^s(\R)) \cap C^1((0,\tau_0],H^{s-1}(\R))
\end{equation}
and $A(0,\cdot) = A_0$. Moreover, the local solution depends continuously on the initial data $A_0$.
\end{proposition}

\begin{remark}
Without additional constraints on $A_0 \in H^s(\R)$, $A_{\tau}$ is not continuous at $\tau = 0$ and will
generally violate the bound (\ref{condition-solution}) required in Theorem \ref{theorem-main}.
\end{remark}

We will need some estimates on the higher derivatives of the local solution $A$ with respect to $\tau$.
Applying the anti-derivative $\partial_{\xi}^{-1}$ to locally integrable functions
in the distribution sense, we obtain from the short-pulse equation (\ref{short-pulse}),
\begin{eqnarray}
\label{first-der}
A_{\tau} & = & \partial_{\xi}^{-1} A + (A^3)_{\xi}, \\
\label{second-der}
A_{\tau \tau} & = & \partial_{\xi}^{-2} A + 3 (A^2)_{\xi} \partial_{\xi}^{-1} A + 4 A^3 + \frac{9}{5} (A^5)_{\xi \xi}, \\
\nonumber
A_{\tau \tau \tau} & = & \partial_{\xi}^{-3} A + \partial_{\xi}^{-1} A^3
+ 18 A^2 \partial_{\xi}^{-1} A + 3 (A^2)_{\xi} \partial_{\xi}^{-2} A
+ 6 A_{\xi} (\partial_{\xi}^{-1} A)^2 \\
\label{third-der}
& \phantom{t} & \phantom{texttt} + \frac{27}{2} ( A^4 )_{\xi \xi} \partial_{\xi}^{-1} A
+ \frac{123}{5} (A^5)_{\xi} + \frac{27}{7} (A^7)_{\xi \xi \xi},
\end{eqnarray}
This chain of equations shows that the derivatives of
the local solution $A$ in $\tau$ can be controlled if the anti-derivatives of $A$ in $\xi$
are controlled. The following lemma gives an useful result for this purpose.

\begin{lemma}
\label{lemma-inhomogeneous-short}
Let $B_0 \in L^2(\R)$ and either (a) $F = G_{\xi}$ with $G \in C([0,\tau_0],L^2(\R))$ or
(b) $F \in C^1([0,\tau_0],L^2(\R))$ for some $\tau_0 > 0$.
The linear inhomogeneous short-pulse  equation,
\begin{equation}
\label{lin-short-pulse}
\left. \begin{array}{l} B_{\tau \xi} = B + F, \\
B(0,\cdot) = B_0, \end{array} \right\}
\end{equation}
admits a unique solution $B \in C([0,\tau_0],L^2(\R))$.
\end{lemma}

\begin{proof}
Let $S(\tau) = e^{\tau \partial_{\xi}^{-1}} : L^2(\R) \to L^2(\R)$ denote
the fundamental solution operator associated with the linear short-pulse
equation $B_{\tau \xi} = B$.
 Using the Fourier transform, we see that the operator $S(\tau)$ is
norm-preserving for any $\tau \in \R$ in the sense $\| S(\tau) B_0 \|_{L^2} = \| B_0 \|_{L^2}$
for any $B_0 \in L^2(\R)$.

In case (a), we rewrite (\ref{lin-short-pulse}) in the integral form,
\begin{equation}
\label{int-short-pulse-a}
B(\tau,\cdot) = S(\tau) B_0 + \int_0^{\tau} S(\tau - \tau') G(\tau',\cdot) d\tau'.
\end{equation}
From the norm-preserving property of $S(\tau)$ and the assumption on $G$ in (a),
we obtain a unique solution $B \in C([0,\tau_0],L^2(\R))$.

In case (b), using the decomposition $B = -F + \tilde{B}$, we rewrite
the initial-value problem (\ref{lin-short-pulse}) in the equivalent form,
\begin{equation}
\label{lin-short-pulse-equivalent}
\left. \begin{array}{l} \tilde{B}_{\tau} = \partial_{\xi}^{-1} \tilde{B} +
F_{\tau}, \\
\tilde{B}(0,\cdot) = \tilde{B}_0, \end{array} \right\}
\end{equation}
where $\tilde{B}_0 = B_0 + F(0,\cdot) \in L^2(\R)$. By Duhamel's principle,
the initial-value problem (\ref{lin-short-pulse-equivalent}) can be written in the integral form,
\begin{equation}
\label{int-short-pulse}
\tilde{B}(\tau,\cdot) = S(\tau) \tilde{B}_0 + \int_0^{\tau} S(\tau - \tau') F_{\tau}(\tau',\cdot) d\tau'.
\end{equation}
From the norm-preserving property of $S(\tau)$ and the assumption on $F$ in (b),
we obtain a unique solution $\tilde{B} \in C([0,\tau_0],L^2(\R))$ and hence
the assertion of the lemma.
\end{proof}

We shall now use Lemma \ref{lemma-inhomogeneous-short}
to control the anti-derivatives of the local solution $A$ in $\xi$.

\begin{corollary}
\label{corollary-first-antider}
Fix $s > \frac{3}{2}$. If $A_0 \in H^s(\R) \cap \dot{H}^{-1}(\R)$, then the local solution
of Proposition \ref{lemma-existence-short} satisfies
\begin{equation}
\label{property-first}
\partial_{\xi}^{-1} A \in C([0,\tau_0],H^{s+1}(\R)), \quad A \in C^1([0,\tau_0],H^{s-1}(\R)).
\end{equation}
\end{corollary}

\begin{proof}
Because $A \in C([0,\tau_0],H^s(\R))$ from Proposition \ref{lemma-existence-short}, we only need to prove
that $\partial_{\xi}^{-1} A \in C([0,\tau_0],L^2(\R))$ in order to show
that $\partial_{\xi}^{-1} A \in C([0,\tau_0],H^{s+1}(\R))$. Then,
$A \in C^1([0,\tau_0],H^{s-1}(\R))$ from equation (\ref{first-der}).

Let us denote $B^{(1)} := \partial_{\xi}^{-1} A$. From
equation (\ref{first-der}), we can see that it satisfies
$$
B^{(1)}_{\tau \xi} = B^{(1)} + (A^3)_{\xi}.
$$
Recall that $H^s(\R)$ is a Banach algebra with respect to pointwise multiplication
for any $s > \frac{1}{2}$. By Lemma \ref{lemma-inhomogeneous-short} in case (a),
if $B^{(1)}_0 \in L^2(\R)$, then $B^{(1)} \in C([0,\tau_0],L^2(\R))$.
\end{proof}

\begin{corollary}
\label{corollary-second-antider}
Fix $s > \frac{5}{2}$. If $A_0 \in H^s(\R) \cap \dot{H}^{-2}(\R)$, then the local solution
of Proposition \ref{lemma-existence-short} satisfies
\begin{equation}
\label{property-second}
\partial_{\xi}^{-2} A \in C([0,\tau_0],H^{s+2}(\R)), \quad
\partial_{\xi}^{-1} A \in C^1([0,\tau_0],H^s(\R)), \quad A \in C^2([0,\tau_0],H^{s-2}(\R)).
\end{equation}
\end{corollary}

\begin{proof}
Denote $B^{(2)} := \partial_{\xi}^{-2} A$ and compute
$$
B^{(2)}_{\xi \tau} = B^{(1)}_{\tau} = B^{(2)} + A^3.
$$
We note that $F = A^3 \in C^1([0,\tau_0],L^2(\R))$ because
of property (\ref{property-first}).
By Lemma \ref{lemma-inhomogeneous-short} in case (b), if $B^{(2)}_0 \in L^2(\R)$, then
$B^{(2)} \in C([0,\tau_0],L^2(\R))$. Hence $\partial_{\xi}^{-2} A \in C([0,\tau_0],H^{s+2}(\R))$
and $\partial_{\xi}^{-1} A \in C^1([0,\tau_0],H^s(\R))$. Then, $A \in C^2([0,\tau_0],H^{s-2}(\R))$
follows from property (\ref{property-first}) and equation (\ref{second-der}).
\end{proof}

\begin{corollary}
\label{corollary-third-antider}
Fix $s > \frac{7}{2}$. If $A_0 \in H^s(\R) \cap \dot{H}^{-2}(\R)$
and $\partial_{\xi}^{-3} A_0 + \partial_{\xi}^{-1} A_0^3 \in L^2(\R)$,
then the local solution of Proposition \ref{lemma-existence-short} satisfies
\begin{equation}
\label{property-third}
A \in C^3([0,\tau_0],H^{s-3}(\R))
\end{equation}
\end{corollary}

\begin{proof}
Denote $B^{(3)} := \partial_{\xi}^{-3} A + \partial_{\xi}^{-1} A^3$ and compute
$$
B^{(3)}_{\xi \tau} = B^{(3)} + 3 A^2 \partial_{\xi}^{-1} A + 9 A^4 A_{\xi}.
$$
We note that $F = 3 A^2 \partial_{\xi}^{-1} A + 9 A^4 A_{\xi} \in C^1([0,\tau_0],L^2(\R))$ because
of property (\ref{property-second}).
By Lemma \ref{lemma-inhomogeneous-short} in case (b), if
$B^{(3)}_0 \in L^2(\R)$, then $B^{(3)} \in C([0,\tau_0],L^2(\R))$.
Then, $A \in C^3([0,\tau_0],H^{s-3}(\R))$ follows from property (\ref{property-second})
and equation (\ref{third-der}).
\end{proof}

Small-norm solutions are known to exist for infinite time of the short-pulse equation.
This result was originally proved in $H^2$ \cite[Theorem 1]{PS10}. Using the blow-up
alternative for the short-pulse equation \cite[Lemma 2]{LPS09}, one can extend this result
to any $s \geq 2$. To be precise, we have the following result.

\begin{proposition}\cite{LPS09,PS10}
\label{lemma-global-solution}
Fix $s \geq 2$. If $A_0 \in H^s(\R)$ and
\begin{equation}
\label{global-solution}
\| A_0' \|_{L^2}^2 + \| A_0'' \|_{L^2}^2 < \frac{1}{6},
\end{equation}
the maximal existence time of the local solution of Proposition \ref{lemma-existence-short} extends to infinity.
Moreover, there exists $C > 0$ such that the unique solution $A \in C(\R_+,H^s(\R))$
of the short-pulse equation (\ref{short-pulse}) with $A(0,\cdot) = A_0$
satisfies $\| A(\tau,\cdot) \|_{H^s} \leq C$ for all $\tau \in \R_+$.
\end{proposition}

\begin{remark}
In the justification result of Theorem \ref{theorem-main},
we need small-norm solutions to control linear error terms.
On the other hand, we do not need continuation of these solutions to infinite time because
the justification analysis only holds on finite time intervals in $\tau$.
\end{remark}

\section{Local solutions of the quasilinear Klein--Gordon equation}

Local well-posedness of the quasi-linear equations was studied by Kato \cite{Kato}.
To employ his formalism, we shall rewrite the quasilinear Klein--Gordon equation (\ref{Maxwell})
as a system of first-order quasi-linear equations with a symmetric matrix. Because solutions of
the short-pulse equation are small solutions of the quasilinear Klein--Gordon equation in the
$L^{\infty}$ norm, we can assume that $\| u \|_{L^{\infty}} < \frac{1}{\sqrt{3}}$
and write
\begin{equation}
u_1 = u_t, \quad u_2 = (1 - 3 u^2)^{1/2} u_x, \quad u_3 = u.
\end{equation}
The quasilinear Klein--Gordon equation (\ref{Maxwell}) is equivalent to the system of
first-order quasi-linear equations,
\begin{equation}
\label{system-quasi-linear}
\frac{\partial}{\partial t} \left[ \begin{array}{c} u_1 \\ u_2 \\ u_3 \end{array} \right] +
\left[ \begin{array}{ccc} 0 & -(1 - 3 u_3^2)^{1/2} & 0 \\
-(1 - 3 u_3^2)^{1/2} & 0 & 0 \\ 0 & 0 & 0 \end{array} \right] \frac{\partial}{\partial x}
\left[ \begin{array}{c} u_1 \\ u_2 \\ u_3 \end{array} \right] =
\left[ \displaystyle{\begin{array}{c} -u_3 - \frac{3 u_2^2 u_3}{1 - 3 u_3^2} \\ - \frac{3 u_1 u_2 u_3}{1 - 3 u_3^2} \\ u_1 \end{array}} \right].
\end{equation}
By Theorems II and III in \cite{Kato}, a unique local solution of
system (\ref{system-quasi-linear}) for the vector $(u_1,u_2,u_3)$ exists in space
$C([0,t_0],H^s(\R)) \cap C^1([0,t_0],H^{s-1}(\R)$  for some $t_0 > 0$ and $s > \frac{3}{2}$. Coming back
to the quasilinear Klein--Gordon equation (\ref{Maxwell}), this result is formulated as follows.

\begin{proposition}\cite{Kato}
\label{lemma-existence-Maxwell}
Fix $s > \frac{3}{2}$. For any $u_0 \in H^{s+1}(\R)$ and $v_0 \in H^s(\R)$
such that $\| u_0 \|_{L^{\infty}} < \frac{1}{\sqrt{3}}$,
there exists a time $t_0 = t_0(\| u_0 \|_{H^{s+1}}+\|v_0\|_{H^s}) > 0$ and
a unique strong solution of the quasilinear Klein--Gordon equation (\ref{Maxwell}) such that
\begin{equation}
\label{Kato-solution}
u \in C([0,t_0],H^{s+1}(\R)) \cap C^1([0,t_0],H^s(\R)) \cap C^2([0,t_0],H^{s-1}(\R)),
\end{equation}
subject to the initial data $u(0,\cdot) = u_0$ and $u_t(0,\cdot) = v_0$.
Moreover, the local solution depends continuously on the initial data $(u_0,v_0)$.
\end{proposition}

Since the existence time $t_0$ may depend on the initial norm $\| u_0 \|_{H^{s+1}} + \| v_0 \|_{H^s}$,
it may be difficult to continue the local solution for infinite time if the norms increase
along the local solution. In some cases, blow-up in a finite time is possible in the $H^{s+1}$ norm
for $u(t,\cdot)$. By the main result (a--ii) of Yin \cite[Theorem 2.3]{Yin}, if
the blow-up occurs in a finite time, it occurs simultaneously in all
$H^{s+1}$ norms for any $s > \frac{3}{2}$. This result is formulated as follows.

\begin{proposition}\cite{Yin}
\label{lemma-Kato-blowup}
The maximal existence time for the local solution in Proposition \ref{lemma-existence-Maxwell} is independent of
$s > \frac{3}{2}$ in the following sense. If two local solutions of the quasilinear Klein--Gordon equation
(\ref{Maxwell}) exist
\begin{equation}
\label{Kato-solution-1}
u \in C([0,t_1),H^{s_1+1}(\R)) \cap C^1([0,t_1),H^{s_1}(\R)) \cap C^2([0,t_1),H^{s_1-1}(\R)),
\end{equation}
and
\begin{equation}
\label{Kato-solution-2}
u \in C([0,t_2),H^{s_2+1}(\R)) \cap C^1([0,t_2),H^{s_2}(\R)) \cap C^2([0,t_2),H^{s_2-1}(\R)),
\end{equation}
for the same initial data $u_0 \in H^{s_1+1}(\R) \cap H^{s_2+1}(\R)$
and $v_0 \in H^{s_1}(\R) \cap H^{s_2}(\R)$ with $s_1, s_2 > \frac{3}{2}$ and $s_1 \neq s_2$,
then $t_1 = t_2$.
\end{proposition}

Results of Propositions \ref{lemma-existence-Maxwell} and \ref{lemma-Kato-blowup}
are useful to establish the criterion that controls the breakdown of local solutions
for the quasilinear Klein--Gordon equation (\ref{Maxwell}). The following lemma
gives the continuation criterion.

\begin{lemma}
\label{lemma-Kato-criterion}
The local solution of the quasilinear Klein--Gordon equation (\ref{Maxwell})
in Proposition \ref{lemma-existence-Maxwell} is continued on the time
interval $[0,t_0]$ for some $t_0 > 0$ as long as
\begin{equation}
\label{Kato-criterion}
\sup_{t \in [0,t_0]} \| u(t,\cdot) \|_{L^{\infty}} < \frac{1}{\sqrt{3}}, \quad
\mbox{\rm and} \quad  \sup_{t \in [0,t_0]} \left( \| u_t(t,\cdot) \|_{L^{\infty}} +
\| u_x(t,\cdot) \|_{L^{\infty}} \right) < \infty.
\end{equation}
\end{lemma}

\begin{proof}
We will prove that the local solution of Proposition \ref{lemma-existence-Maxwell}
does not blow up in the $H^s$ norm on the time interval $[0,t_0]$ if
\begin{equation}
\label{contradiction}
M_0 < \frac{1}{\sqrt{3}} \quad \mbox{\rm and} \quad M_1 + M_2 < \infty,
\end{equation}
where
$$
M_0 = \sup_{t \in [0,t_0]} \| u(t,\cdot) \|_{L^{\infty}}, \quad
M_1 =  \sup_{t \in [0,t_0]}\| u_t(t,\cdot) \|_{L^{\infty}}, \quad
M_2 = \sup_{t \in [0,t_0]} \| u_x(t,\cdot) \|_{L^{\infty}}.
$$
Because of the independence of the blow-up time from the index $s$ in Proposition \ref{lemma-Kato-blowup},
it suffices to consider the simplest $H^{s+1}$ norm for $u$ with $s = 2 > \frac{3}{2}$.

Let us define the sequence of energies for the quasilinear Klein--Gordon equation (\ref{Maxwell}),
\begin{eqnarray}
E_1(u) & = & \int_{\R} (u^2 + u_t^2 + u_x^2 (1 - 3 u^2)) dx, \\
E_2(u) & = & \int_{\R} (u_x^2 + u_{tx}^2 + u_{xx}^2 (1 - 3 u^2)) dx, \\
E_3(u) & = & \int_{\R} (u_{xx}^2 + u_{txx}^2 + u_{xxx}^2 (1 - 3 u^2)) dx.
\end{eqnarray}

Multiplying equation (\ref{Maxwell}) by $u_t$, we obtain the energy balance equation,
\begin{equation}
\frac{1}{2} \frac{d E_1(u)}{d t} = - 3 \int_{\R} u u_t u_x^2 dx, \quad t \in [0,t_0],
\end{equation}
where the decay of $u, u_t, u_x$ to $0$ as $|x| \to \infty$ is used. This decay is justified
for any local solution of Proposition \ref{lemma-existence-Maxwell}. Under the assumption (\ref{contradiction}),
there is $C(M_0) > 0$ such that
\begin{equation}
\left| \frac{d E_1(u)}{d t} \right| \leq C(M_0) M_0 M_1 E_1(u) \quad \Rightarrow
\quad E_1(u) \leq E_1(u_0) e^{C(M_0) M_0 M_1 t}, \quad t \in [0,t_0].
\end{equation}
Therefore, $E_1(u)$ cannot blow up on the time interval $[0,t_0]$ if
$M_0 < \frac{1}{\sqrt{3}}$ and $M_1 < \infty$.

Differentiating equation (\ref{Maxwell}) in $x$ and multiplying the resulting equation by $u_{tx}$, we obtain the energy balance equation,
\begin{equation}
\frac{1}{2} \frac{d E_2(u)}{d t} = - 3 \int_{\R} u u_t u_{xx}^2 dx - 6 \int_{\R} u_x^3 u_{tx} dx - 12 \int_{\R} u u_x u_{xx} u_{tx} dx, \quad t \in [0,t_0],
\end{equation}
where the decay of $u_{tx}, u_{xx}$ to $0$ as $|x| \to \infty$ is used.
Again, this decay is justified for any local solution of
Proposition \ref{lemma-existence-Maxwell}. Under the assumption (\ref{contradiction}),
there is $C(M_0) > 0$ such that
\begin{eqnarray}
\nonumber
& \phantom{t} & \left| \frac{d E_2(u)}{d t} \right| \leq C(M_0) (M_0 M_1 + 2 M_2^2 + 4 M_0 M_2) E_2(u) \\
& \phantom{t} & \phantom{texttext} \Rightarrow \quad
E_2(u) \leq E_2(u_0) e^{C(M_0) (M_0 M_1 + 2 M_2^2 + 4 M_0 M_2) t}, \quad t \in [0,t_0].
\end{eqnarray}
Therefore, $E_2(u)$ cannot blow up on the time interval $[0,t_0]$ if
$M_0 < \frac{1}{\sqrt{3}}$ and $M_1 + M_2 < \infty$.

We need one more computation for $E_3(u)$ to control the $H^3$ norm for solution $u$. However,
because of the integration over $x \in \R$, we can not work directly with the local solution $u$
and need the approximating sequence $\{ u^{(n)} \}_{n \in \mathbb{N}}$ of local solutions in Sobolev
space of a higher index $s = 3 > \frac{3}{2}$. Applying density arguments and continuous dependence
from initial data, we approximate the initial value $u_0 \in H^3(\R)$ and $v_0 \in H^2(\R)$ by functions
$u_0^{(n)} \in H^4(\R)$ and $v_0^{(n)} \in H^3(\R)$, such that $u_0^{(n)} \to u_0$ in $H^3$ and
$v_0^{(n)} \to v_0$ in $H^2$ as $n \to \infty$.
The approximating sequence $\{ u^{(n)} \}_{n \in \mathbb{N}}$ of
local solutions of the quasilinear Klein--Gordon equation (\ref{Maxwell}) is generated by the sequence
of the initial data $\{ u_0^{(n)}\}_{n \in \mathbb{N}}$ and $\{v_0^{(n)}\}_{n \in \mathbb{N}}$.

Differentiating equation (\ref{Maxwell}) twice in $x$ and multiplying the resulting
equation by $u_{txx}$, we obtain the energy balance equation,
\begin{eqnarray}
\nonumber
\frac{1}{2} \frac{d E_3(u^{(n)})}{d t} & = & - 3 \int_{\R} u^{(n)} u^{(n)}_t  (u^{(n)}_{xxx})^2 dx
- 36 \int_{\R}  (u^{(n)}_x)^2  u^{(n)}_{xx}  u^{(n)}_{txx} dx - 18 \int_{\R}  u^{(n)}  u^{(n)}_x  u^{(n)}_{xxx}  u^{(n)}_{txx} dx \\
& \phantom{t} & - 18 \int_{\R}  u^{(n)}  (u^{(n)}_{xx})^2  u^{(n)}_{txx} dx, \qquad
\qquad \qquad t \in [0,t_0],\label{last-term}
\end{eqnarray}
where the decay of $ u^{(n)}_{txx},  u^{(n)}_{xxx}$ to $0$ as $|x| \to \infty$ is used.
 This decay is justified for the approximating sequence
$\{ u^{(n)} \}_{n \in \mathbb{N}}$ of local solutions of Proposition
\ref{lemma-existence-Maxwell} with $s = 3$.
Under the assumption (\ref{contradiction}) for
the approximating sequence
$\{ u^{(n)} \}_{n \in \mathbb{N}}$ rewritten as $M_0^{(n)} < \frac{1}{\sqrt{3}}$ and
$M_1^{(n)} + M_2^{(n)} < \infty$, there is $C(M_0^{(n)}) > 0$ such that
\begin{equation}
\left| \frac{d E_3(u^{(n)})}{d t} \right| \leq C(M_0^{(n)}) (M_0^{(n)} M_1^{(n)} + 12 (M_2^{(n)})^2 +
9 M_0^{(n)} M_2^{(n)} + 9 M_0^{(n)} E_2^{1/2}(u^{(n)})) E_3(u^{(n)}),
\end{equation}
where the Gagliardo--Nirenberg inequality is used to estimate the last term of (\ref{last-term}),
\begin{eqnarray*}
\left| \int_{\R}  u^{(n)}  (u^{(n)}_{xx})^2  u^{(n)}_{txx} dx \right| & \leq &
M_0^{(n)} \| u^{(n)}_{txx} \|_{L^2} \| u^{(n)}_{xx} \|_{L^4}^2 \\
& \leq & M_0^{(n)} \| u_{txx}^{(n)} \|_{L^2} \| u_{xxx}^{(n)} \|_{L^2}^{1/2} \| u_{xx} \|_{L^2}^{3/2} \\
& \leq & C(M_0^{(n)}) E_2^{1/2}(u^{(n)}) E_3(u^{(n)}).
\end{eqnarray*}
Since $E_3(u^{(n)}_0) \to E_3(u_0)$ as $n \to \infty$, we infer from the continuous dependence
of the local solution $u$ on initial data $u_0$ that
$E_3(u)$ cannot blow up on the time interval $[0,t_0]$ if
$M_0 < \frac{1}{\sqrt{3}}$ and $M_1 + M_2 < \infty$.
\end{proof}

\begin{remark}
The continuation criterion of Lemma \ref{lemma-Kato-criterion} will allow us to prove
the estimates for the approximation in the $H^2$-norm and to avoid energy
estimates in higher Sobolev spaces. In addition, it will allow us to extend the local solution of the quasilinear
Klein--Gordon equation (\ref{Maxwell}) to the times $t = \mathcal{O}(1/\epsilon)$
or $T = \mathcal{O}(1)$ required for the justification result in Theorem \ref{theorem-main}.
\end{remark}

The results of Proposition \ref{lemma-existence-Maxwell} and Lemma \ref{lemma-Kato-criterion}
can now be rewritten for the equivalent quasilinear Klein--Gordon equation (\ref{Maxwell-short})
in new variables (\ref{new-variables}).

\begin{corollary}
\label{corollary-existence-Maxwell}
Fix $s > \frac{3}{2}$ and $C_0 > 0$ independently of $\epsilon$.
For any $U_0 \in H^{s+1}(\R)$ and $V_0 \in H^s(\R)$
such that $\| U_0 \|_{L^{\infty}} \leq C_0$ and for all small
nonzero $\epsilon$, there exists an $\epsilon$-independent
time $T = T(\| U_0 \|_{H^{s+1}}+\|V_0\|_{H^s}) > 0$ and
a unique strong solution of the quasilinear Klein--Gordon equation (\ref{Maxwell-short})
such that
\begin{equation}
\label{Kato-solution-short}
U(\tau,\cdot) \in C([0,\epsilon T],H^{s+1}(\R)) \cap C^1([0,\epsilon T],H^s(\R)) \cap C^2([0,\epsilon T],H^{s-1}(\R)),
\end{equation}
subject to the initial data $U(0,\cdot) = U_0$ and $U_{\tau}(0,\cdot) = V_0$. Moreover,
for all small nonzero $\epsilon$, the local solution
is continued on the time interval $[0,\tau_0]$ for a $\tau_0 > 0$ as long as
there is an $\epsilon$-independent positive constant $C$ such that
\begin{equation}
\label{Kato-criterion-short}
\sup_{\tau \in [0,\tau_0]} \| U(\tau,\cdot) \|_{L^{\infty}} \leq C \quad
\mbox{\rm and} \quad \sup_{\tau \in [0,\tau_0]} \left( \| U_{\tau}(\tau,\cdot) \|_{L^{\infty}} +
\| U_{\xi}(\tau,\cdot) \|_{L^{\infty}} \right) < \infty.
\end{equation}
\end{corollary}

\begin{proof}
The result follows by the transformation of variables (\ref{new-variables}).
\end{proof}

\section{Energy estimates for the error term}

To continue with the justification analysis, we decompose a solution of 
the quasilinear Klein--Gordon equation (\ref{Maxwell-short})
in the form $U = A + \epsilon R$, where
$A$ is a solution of the short-pulse equation (\ref{short-pulse}) 
satisfying condition (\ref{condition-solution}) of Theorem \ref{theorem-main} 
and $R$ is the error term satisfying
\begin{equation}
\label{remainder}
R_{\xi \tau} = R + \epsilon^2 R_{\tau \tau}
+ \left( 3 A^2 R + 3 \epsilon A R^2 + \epsilon^2 R^3 \right)_{\xi \xi} + \epsilon A_{\tau\tau}.
\end{equation}
We shall now control solutions of this error equation by using apriori energy estimates.
The energy for the error term is defined by
\begin{equation}
\label{energy-R}
E = \int_{\R} \left( R^2 + R_{\xi}^2 + R_{\xi \xi}^2 + 2 \epsilon^2 R_{\tau}^2 + \epsilon^4 R_{\tau \tau}^2 \right) dx.
\end{equation}

By Sobolev embedding, the energy space $E < \infty$ is embedded into the space of continuously differentiable
functions in $\xi$ on $\R$, which are decaying to zero at infinity as $|\xi| \to \infty$ and are bounded by
\begin{equation}
\| R \|_{L^{\infty}} + \| R_{\xi} \|_{L^{\infty}} \leq C E^{1/2}
\end{equation}
In addition, if $A$ is a strong solution of the short-pulse equation from Proposition \ref{lemma-existence-short}
satisfying property (\ref{property-second}) of Corollary \ref{corollary-second-antider}
and $E < \infty$, then $R_{\xi \tau}$ defined by equation (\ref{remainder}) belongs to $L^2(\R)$ such that
\begin{eqnarray*}
\| R_{\xi \tau} \|_{L^2} & \leq & \epsilon \|A_{\tau \tau} \|_{L^2} + \| R \|_{L^2} + \epsilon^2 \|R_{\tau \tau} \|_{L^2}
+ 3 (\| A \|_{L^{\infty}} + \epsilon \| R \|_{L^{\infty}})^2 \| R_{\xi \xi} \|_{L^2} \\
& \phantom{t} & + 6 \epsilon (\| A \|_{L^{\infty}} + \epsilon \| R \|_{L^{\infty}}) \| R_{\xi} \|^2_{L^2}
+ 12 (\| A \|_{L^{\infty}} + \epsilon \| R \|_{L^{\infty}}) \| A_{\xi} \|_{L^2} \| R_{\xi} \|_{L^2} \\
& \phantom{t} & + 3 (2 \|A \|_{L^{\infty}} \| R \|_{L^{\infty}} + \epsilon \| R \|^2_{L^{\infty}} ) \| A_{\xi \xi} \|_{L^2}
+ 6 \|R \|_{L^{\infty}} \|A_{\xi} \|_{L^2}^2.
\end{eqnarray*}
The previous lengthy estimate can be greatly simplified if $R$ belongs to the energy space
(\ref{energy-R}) and $A$ satisfies condition (\ref{condition-solution}) of Theorem \ref{theorem-main}. In
this case, for sufficiently small $\epsilon > 0$, we write
\begin{equation}
\label{bound-R-xi-tau}
\| R_{\xi \tau} \|_{L^2} \leq C \left( \delta \epsilon + E^{1/2} +
\delta^2 E^{1/2} + \delta \epsilon E + \epsilon^2 E^{3/2}  \right),
\end{equation}
where $C$ is a generic $\epsilon$-independent positive constant, which may change from
one line to another line and from one term to another term in the same inequality.

By Sobolev's embedding, (\ref{energy-R}) and (\ref{bound-R-xi-tau}) yield
the control of $R_{\tau}$ in $L^{\infty}$ norm with the bound
\begin{equation}
\label{bound-R-tau-infty}
\| \epsilon R_{\tau} \|_{L^{\infty}} \leq C \left( \| \epsilon R_{\tau} \|_{L^2}
+ \epsilon \| R_{\xi \tau} \|_{L^2} \right) \leq C \left( E^{1/2} + \delta \epsilon^2
+ \delta \epsilon^2 E + \epsilon^3 E^{3/2} \right),
\end{equation}
where we have used that $(1 + \epsilon + \epsilon \delta^2) E^{1/2} \leq C E^{1/2}$.
Note that $R_{\tau}$ is a continuous function of $\xi$, which decays to zero
at infinity as $|\xi| \to \infty$.

The main result of this section is the following lemma.

\begin{lemma}
\label{lemma-remainder}
Under the assumptions of Theorem \ref{theorem-main}, the rate of change of
the energy (\ref{energy-R}) is given by
\begin{equation}
\label{balance-energy}
\frac{d}{d \tau} \left( E + \tilde{E} \right) = J,
\end{equation}
where $E$ is given by (\ref{energy-R}), $\tilde{E}$ is given by
\begin{eqnarray*}
\tilde{E} & = & \int_{\R}  \left( - 2 \epsilon^2 R_{\xi} R_{\tau} - 3 A^2 R_{\xi}^2
- 6 \epsilon A R R_{\xi}^2 - 3 \epsilon^2 R^2 R_{\xi}^2 \right) d \xi \\
& \phantom{t} & + \int_{\R}  \left( - 2 \epsilon^2 R_{\xi \xi} R_{\xi \tau}
- 3 \epsilon^2 A^2 R_{\xi \tau}^2 + 6 \epsilon^2 (A A_{\xi})_{\xi} R_{\tau}^2 - 6 \epsilon^3 A R R_{\xi \tau}^2
- 3 \epsilon^4 R^2 R_{\xi \tau}^2 \right) d \xi
\end{eqnarray*}
and $J$ is given by the sum of the right-hand-sides of (\ref{balance-H1-norm}) and (\ref{balance-H2-norm}) below.
Moreover, for sufficiently small $\delta > 0$ and $\epsilon > 0$, there is
an ($\delta$,$\epsilon$)-independent constant $C > 0$ such that
\begin{eqnarray}
\label{balance-1}
|\tilde{E}| & \leq & C \left( \epsilon E + \delta^2 E + \delta \epsilon E^{3/2} + \epsilon^2 E^2 \right), \\
\label{balance-2}
|J| & \leq & C \left( \delta E^{1/2} + \delta^2 E + \delta E^{3/2} + \epsilon E^2 \right),
\end{eqnarray}
as long as the solution remains in the function space
\begin{equation}
\label{Kato-solution-long}
R \in C([0,T],H^3(\R)) \cap C^1([0,T],H^2(\R)) \cap C^2([0,T],H^1(\R)).
\end{equation}
\end{lemma}

The proof of Lemma \ref{lemma-remainder} is based on a number of elementary but lengthy computations.
Multiplying equation (\ref{remainder}) by $R_{\xi}$, we derive the first balance equation,
\begin{eqnarray}
\nonumber
& \phantom{t} &  \frac{\partial}{\partial \tau} \left( -\frac{1}{2} R_{\xi}^2 + \epsilon^2 R_{\xi} R_{\tau} \right) \\
\nonumber & \phantom{t} &
+ \frac{\partial}{\partial \xi} \left( \frac{1}{2} R^2 - \frac{1}{2} \epsilon^2 R_{\tau}^2
+ \frac{3}{2} A^2 R_{\xi}^2 + \frac{3}{2} (A^2)_{\xi \xi} R^2 + \epsilon A_{\xi \xi} R^3
+ 3 \epsilon A R R_{\xi}^2 + \frac{3}{2} \epsilon^2 R^2 R^2_{\xi} \right) \\ \label{balance-xi}
& = & - \epsilon R_{\xi} A_{\tau \tau} - 3 A A_{\xi} R_{\xi}^2 + \frac{3}{2} (A^2)_{\xi \xi \xi} R^2
+ \epsilon A_{\xi \xi \xi} R^3 - 9 \epsilon A_{\xi} R R_{\xi}^2 - 3 \epsilon A R_{\xi}^3
- 3 \epsilon^2 R R_{\xi}^3.
\end{eqnarray}
Multiplying equation (\ref{remainder}) by $R_{\tau}$, we derive the second balance equation,
\begin{eqnarray}
\nonumber
& \phantom{t} &  \frac{\partial}{\partial \tau} \left( \frac{1}{2} R^2 + \frac{1}{2} \epsilon^2 R^2_{\tau}
- \frac{3}{2} A^2 R_{\xi}^2 - 3 \epsilon A R R_{\xi}^2 - \frac{3}{2} \epsilon^2 R^2 R_{\xi}^2 \right) \\
\nonumber & \phantom{t} &
+ \frac{\partial}{\partial \xi} \left( -\frac{1}{2} R_{\tau}^2 + 3 A^2 R_{\xi} R_{\tau}
+ 6 \epsilon A R R_{\xi} R_{\tau} + 3 \epsilon^2 R^2 R_{\xi} R_{\tau} \right) \\
\nonumber
& = & - \epsilon R_{\tau} A_{\tau \tau} - 3 A A_{\tau} R_{\xi}^2 - 3 (A^2)_{\xi \xi} R R_{\tau}
- 6 A A_{\xi} R_{\xi} R_{\tau} \\ \label{balance-tau}
& \phantom{t} & - 3 \epsilon A_{\xi \xi} R^2 R_{\tau} - 6 \epsilon A_{\xi} R R_{\xi} R_{\tau}
- 3 \epsilon A_{\tau} R R_{\xi}^2 - 3 \epsilon A R_{\xi}^2 R_{\tau} - 3 \epsilon^2 R R_{\xi}^2 R_{\tau}.
\end{eqnarray}

If $R$ belongs to the energy space $E < \infty$, we can integrate
the balance equations (\ref{balance-xi}) and (\ref{balance-tau}) over $\xi$ in $\R$
and use the decay of $R$, $R_{\xi}$, and $R_{\tau}$ to zero at infinity as $\xi \to \infty$.
As a result, we obtain the energy balance equation,
\begin{eqnarray}
\nonumber
& \phantom{t} & \frac{d}{d \tau} \int_{\R}  \left( R^2 + \epsilon^2 R^2_{\tau}
+ R_{\xi}^2 - 2 \epsilon^2 R_{\xi} R_{\tau} - 3  A^2 R_{\xi}^2
- 6 \epsilon A R R_{\xi}^2 - 3 \epsilon^2 R^2 R_{\xi}^2 \right) d \xi \\
\nonumber
& = & 2 \epsilon \int_{\R} (R_{\xi} - R_{\tau}) A_{\tau \tau} d\xi
+ 6 \int_{\R} \left( A A_{\xi} R_{\xi}^2 - (A A_{\xi})_{\xi \xi} R^2 -
A A_{\tau} R_{\xi}^2 + 2 A A_{\xi} R R_{\xi \tau} \right) d \xi \\ \nonumber
& \phantom{t} & + 2 \epsilon \int_{\R} \left(
- A_{\xi \xi \xi} R^3 + 9 A_{\xi} R R_{\xi}^2 + 3 A R_{\xi}^3
- 3 A_{\xi \xi} R^2 R_{\tau} - 6 A_{\xi} R R_{\xi} R_{\tau}
- 3 A_{\tau} R R_{\xi}^2 - 3 A R_{\xi}^2 R_{\tau} \right) d \xi \\
\label{balance-H1-norm}
& \phantom{t} & + 6 \epsilon^2 \int_{\R} R R_{\xi}^2 \left( R_{\xi} - R_{\tau} \right) d\xi,
\end{eqnarray}
where the integration by parts is performed to obtain
$$
\int_{\R} \left( (A^2)_{\xi \xi} R R_{\tau} +  2 A A_{\xi} R_{\xi} R_{\tau} \right) d \xi =
- \int_{R} 2 A A_{\xi} R R_{\tau \xi} d \xi.
$$

We still need estimates of the rate of change of
$\| R_{\xi \xi} \|_{L^2}^2$ and $\| \epsilon^2 R_{\tau \tau} \|_{L^2}^2$.
Taking the derivative of equation (\ref{remainder}) in $\xi$ and
multiplying the resulting equation by $R_{\xi \xi}$, we derive the third balance equation,
\begin{eqnarray}
\nonumber
& \phantom{t} &  \frac{\partial}{\partial \tau} \left( -\frac{1}{2} R_{\xi \xi}^2 + \epsilon^2 R_{\xi \xi} R_{\tau \xi} \right) + \frac{\partial}{\partial \xi} \left( \frac{1}{2} R_{\xi}^2 - \frac{1}{2} \epsilon^2 R_{\tau \xi}^2
+ \frac{3}{2} A^2 R_{\xi \xi}^2 + 3 \epsilon A R R_{\xi \xi}^2
+ \frac{3}{2} \epsilon^2 R^2 R^2_{\xi \xi} + \frac{3}{2} \epsilon^2 R_{\xi}^4 \right) \\ \nonumber
& = & - \epsilon R_{\xi \xi} A_{\tau \tau \xi} - 15 A A_{\xi} R_{\xi \xi}^2
- 18 (A A_{\xi})_{\xi} R_{\xi} R_{\xi \xi} - 6 (A A_{\xi})_{\xi \xi} R R_{\xi \xi}
- 3 \epsilon A_{\xi \xi \xi} R^2 R_{\xi \xi}\\
 \label{balance-xi-xi} &
\phantom{t} & - 18 \epsilon A_{\xi \xi} R R_{\xi} R_{\xi \xi}
- 15 \epsilon A_{\xi} R R_{\xi \xi}^2 - 18 \epsilon A_{\xi} R_{\xi}^2 R_{\xi \xi}
- 15 \epsilon A R_{\xi} R_{\xi \xi}^2 - 15 \epsilon^2 R R_{\xi} R_{\xi \xi}^2.
\end{eqnarray}

Finally, taking the derivative of equation (\ref{remainder}) in $\tau$ and
multiplying the resulting equation by $R_{\tau \tau}$, we derive the last balance equation,
\begin{eqnarray}
\nonumber
& \phantom{t} & \frac{\partial}{\partial \tau} \left( \frac{1}{2} R_{\tau}^2 + \frac{1}{2} \epsilon^2 R^2_{\tau \tau}
- \frac{3}{2} A^2 R_{\xi \tau}^2 + 3 (A A_{\xi})_{\xi} R_{\tau}^2 - 3 \epsilon A R R_{\xi \tau}^2
- \frac{3}{2} \epsilon^2 R^2 R_{\xi \tau}^2 \right) \\
\nonumber &
\phantom{t} & + \frac{\partial}{\partial \xi} \left( -\frac{1}{2} R_{\tau \tau}^2 + 3 A^2 R_{\tau \tau} R_{\xi \tau}
+ 6 \epsilon A R R_{\tau \tau} R_{\xi \tau} + 3 \epsilon^2 R^2 R_{\tau \tau} R_{\xi \tau} \right) \\
\nonumber
& = & - \epsilon R_{\tau \tau} A_{\tau \tau \tau} - 3 A A_{\tau} R_{\xi \tau}^2
- 6 A A_{\xi} R_{\tau \tau} R_{\xi \tau} - 6 A A_{\tau} R_{\tau \tau} R_{\xi \xi}  \\
\nonumber & \phantom{t} &
+ 3 (A A_{\tau})_{\xi \xi} R_{\tau}^2
- 6 (A A_{\tau})_{\xi \xi} R R_{\tau \tau} - 12 (A A_{\tau})_{\xi} R_{\xi} R_{\tau \tau}
- 3 \epsilon A_{\xi \xi \tau} R^2 R_{\tau \tau}\\
\nonumber & \phantom{t} &
- 12 \epsilon A_{\xi \tau} R R_{\xi} R_{\tau \tau} - 6 \epsilon A_{\tau} (R R_{\xi})_{\xi} R_{\tau \tau}
-6 \epsilon A_{\xi \xi} R R_{\tau} R_{\tau \tau} - 12 \epsilon A_{\xi} R_{\xi} R_{\tau} R_{\tau \tau} \\
\nonumber & \phantom{t} &
- 6 \epsilon A_{\xi} R R_{\tau \tau} R_{\xi \tau}
- 3 \epsilon A_{\tau} R R_{\xi \tau}^2 - 6 \epsilon A R_{\tau} R_{\tau \tau} R_{\xi \xi}
- 6 \epsilon A R_{\xi} R_{\tau \tau} R_{\xi \tau} - 3 \epsilon A R_{\tau} R_{\xi \tau}^2 \\
 \label{balance-tau-tau} & \phantom{t} &
- 6 \epsilon^2 R_{\xi}^2 R_{\tau} R_{\tau \tau} - 6 \epsilon^2 R R_{\tau} R_{\xi \xi} R_{\tau \tau}
- 6 \epsilon^2 R R_{\xi} R_{\xi \tau} R_{\tau \tau} - 3 \epsilon^2 R R_{\tau} R_{\xi \tau}^2.
\end{eqnarray}

Let us now assume the decay of $R$, $R_{\xi}$, $R_{\tau}$, $R_{\xi \xi}$, $R_{\tau \xi}$,
$R_{\tau \tau}$ to zero at infinity as $\xi \to \infty$. The decay holds for
the local solution of Corollary \ref{corollary-existence-Maxwell} on the short time
interval $[0, \epsilon T]$, since the assumptions of Theorem  \ref{theorem-main} corresponds
to $s = 2 > \frac{3}{2}$ in Corollary \ref{corollary-existence-Maxwell}.
Integrating the balance equations (\ref{balance-xi-xi}) and
(\ref{balance-tau-tau}) multiplied by $\epsilon^2$ over $\xi$ in $\R$, we obtain
the extended energy balance equation,
\begin{eqnarray}
\nonumber
& \phantom{t} & \frac{d}{d \tau} \int_{\R}  \left( \epsilon^2 R_{\tau}^2 +
\epsilon^4 R^2_{\tau \tau}
+ R_{\xi \xi}^2 - 2 \epsilon^2 R_{\xi \xi} R_{\xi \tau} \right) d \xi \\
\nonumber & \phantom{t} & + \frac{d}{d \tau} \int_{\R}  \left( - 3 \epsilon^2
A^2 R_{\xi \tau}^2 + 6 \epsilon^2 (A A_{\xi})_{\xi} R_{\tau}^2 - 6 \epsilon^3 A R R_{\xi \tau}^2
- 3 \epsilon^4 R^2 R_{\xi \tau}^2 \right) d \xi \\
\label{balance-H2-norm}
& = & 2 \epsilon \int_{\R} (R_{\xi \xi} A_{\tau \tau \xi} - \epsilon^2 R_{\tau \tau} A_{\tau \tau \tau}) d\xi
+ 2 \int_{\R} \left( I_1 + \epsilon I_2 + \epsilon^2 I_3 \right) d \xi,
\end{eqnarray}
where
\begin{eqnarray*}
I_1 & = & 15 A A_{\xi} R_{\xi \xi}^2 + 18 (A A_{\xi})_{\xi} R_{\xi} R_{\xi \xi} +
6 (A A_{\xi})_{\xi \xi} R R_{\xi \xi} - 3 \epsilon^2 A A_{\tau} R_{\xi \tau}^2
- 6 \epsilon^2 A A_{\xi} R_{\tau \tau} R_{\xi \tau}  \\
\nonumber & \phantom{t} &
- 6 \epsilon^2 A A_{\tau} R_{\tau \tau} R_{\xi \xi}  + 3 \epsilon^2 (A A_{\tau})_{\xi \xi} R_{\tau}^2
- 6 \epsilon^2 (A A_{\tau})_{\xi \xi} R R_{\tau \tau} - 12 \epsilon^2 (A A_{\tau})_{\xi} R_{\xi} R_{\tau \tau}
\end{eqnarray*}
\begin{eqnarray*}
I_2 & = & 3 A_{\xi \xi \xi} R^2 R_{\xi \xi} + 18 A_{\xi \xi} R R_{\xi} R_{\xi \xi}
+ 15 A_{\xi} R R_{\xi \xi}^2 + 18 A_{\xi} R_{\xi}^2 R_{\xi \xi}
+ 15 A R_{\xi} R_{\xi \xi}^2 \\
\nonumber & \phantom{t} &- 3 \epsilon^2 A_{\xi \xi \tau} R^2 R_{\tau \tau}
- 12 \epsilon^2 A_{\xi \tau} R R_{\xi} R_{\tau \tau} - 6 \epsilon^2 A_{\tau} (R R_{\xi})_{\xi} R_{\tau \tau}
-6 \epsilon^2 A_{\xi \xi} R R_{\tau} R_{\tau \tau} - 12 \epsilon^2 A_{\xi} R_{\xi} R_{\tau} R_{\tau \tau} \\
\nonumber & \phantom{t} &
- 6 \epsilon^2 A_{\xi} R R_{\tau \tau} R_{\xi \tau}
- 3 \epsilon^2 A_{\tau} R R_{\xi \tau}^2 - 6 \epsilon^2 A R_{\tau} R_{\tau \tau} R_{\xi \xi}
- 6 \epsilon^2 A R_{\xi} R_{\tau \tau} R_{\xi \tau} - 3 \epsilon^2 A R_{\tau} R_{\xi \tau}^2
\end{eqnarray*}
and
\begin{eqnarray*}
I_3 & = & 15 R R_{\xi} R_{\xi \xi}^2 - 6 \epsilon^2 R_{\xi}^2 R_{\tau} R_{\tau \tau} - 6 \epsilon^2 R R_{\tau} R_{\xi \xi} R_{\tau \tau}
- 6 \epsilon^2 R R_{\xi} R_{\xi \tau} R_{\tau \tau} - 3 \epsilon^2 R R_{\tau} R_{\xi \tau}^2.
\end{eqnarray*}
The energy balance equation (\ref{balance-energy})
follows from (\ref{balance-H1-norm}) and (\ref{balance-H2-norm}).
Recall the assumptions on $A$ in Theorem  \ref{theorem-main}.
Using bounds (\ref{condition-solution}), (\ref{bound-R-xi-tau}), and (\ref{bound-R-tau-infty})
together with the Cauchy-Schwarz inequality, we obtain the bounds (\ref{balance-1}), and (\ref{balance-2}).
The proof of Lemma \ref{lemma-remainder} is complete, as long as the local solution
$R$ remain in the class of functions (\ref{Kato-solution-long}).

\section{Continuation arguments and the proof of Theorem \ref{theorem-main}}

We shall now finish the proof of Theorem \ref{theorem-main}.
Assumption (\ref{condition-solution}) is satisfied for a local
solution of the short-pulse equation (\ref{short-pulse})
according to Corollaries \ref{corollary-first-antider},
\ref{corollary-second-antider}, and \ref{corollary-third-antider}
for any fixed $s > \frac{7}{2}$ and $T > 0$.
Assumptions (\ref{error-initial}) after the decomposition $U = A + \epsilon R$
is rewritten in the form,
\begin{equation}
\| R(0,\cdot) \|_{H^2} + \| R_{\tau}(0,\cdot) \|_{H^1} \leq 1.
\end{equation}
This assumption implies that the initial energy $E |_{\tau= 0} < \infty$ and $E |_{\tau = 0} = \mathcal{O}(1)$ as
$\epsilon \to 0$, where the evolution equation (\ref{remainder}) must be used.
Let us denote $E$ at the time $\tau \geq 0$ by $E(\tau)$.

Since $R(0,\cdot) \in H^3(\R)$ and $R_{\tau}(0,\cdot) \in H^2(\R)$ by
the assumption of Theorem \ref{theorem-main}, Corollary \ref{corollary-existence-Maxwell}
with $s = 2$ implies that there exists a local solution
\begin{equation}
\label{Kato-solution-long-short}
R \in C([0,\epsilon T],H^3(\R)) \cap C^1([0,\epsilon T],H^2(\R)) \cap C^2([0,\epsilon T],H^1(\R))
\end{equation}
of the residual equation (\ref{remainder}). Because of the continuation
criterion (\ref{Kato-criterion-short}) in Corollary \ref{corollary-existence-Maxwell}, we
can extend the existence interval to $[0,T]$ as long as $R$ is controlled in the energy space $E(\tau) < \infty$
for $\tau \in [0,T]$.

By Lemma \ref{lemma-remainder}, we have
\begin{equation}
E(\tau) + \tilde{E}(\tau) = E(0) + \tilde{E}(0) + \int_0^{\tau} J(\tau')  d \tau'.
\end{equation}
We use bounds (\ref{balance-1}) and (\ref{balance-2}), the elementary bound
$2 E^{1/2} \leq 1 + E$, and Gronwall's inequality. As a result,
for a sufficiently small $\delta > 0$,
there is a $\epsilon_0 > 0$ such that for all $\epsilon \in (0,\epsilon_0)$,
there are $C_0 > 0$ and $C_1 > 0$ such that
\begin{equation}
\label{last-equation}
E(\tau) \leq C_0 ( E(0) + \delta T ) e^{C_1 \delta T}, \quad \tau \in [0,T].
\end{equation}
Hence, we have $E(\tau) < \infty$ for any $\tau \in [0,T]$, so that the local solution
(\ref{Kato-solution-long-short}) is extended to the whole time interval $[0,T]$.
Because $E(0)$, $T$, $C_0$, and $C_1$ are $\epsilon$-independent,
the proof of Theorem \ref{theorem-main} is complete.


\begin{thebibliography}{99}

\bibitem{AR02} D. Alterman and J. Rauch, ``Nonlinear geometric optics for short pulses",
J. Diff. Eqs. {\bf 178} (2002), 437--465.

\bibitem{AR03} D. Alterman and J. Rauch, ``Diffractive nonlinear geometric optics for short pulses",
SIAM J. Math. Anal. {\bf 34} (2003), 1477-–1502.

\bibitem{Vladimirov} S. Amiranashvili and A. Demircan, ``Hamiltonian structure of propagation equations
for ultrashort optical pulses", Phys. Rev. A {\bf 82} (2010), 013812, 11pp.

\bibitem{BL02} K. Barrailh and D. Lannes, ``A general framework for diffractive optics and its applications to
lasers with large spectrums and short pulses", SIAM J. Math. Anal. {\bf 34} (2002), 636-–674.

\bibitem{Br05} J.C. Brunelli, ``The short pulse hierarchy'',
J. Math. Phys. \textbf{46} (2005), 123507, 9pp.

\bibitem{CJSW05} Y. Chung, C.K.R.T. Jones, T. Sch\"{a}fer, and
C.E. Wayne, ``Ultra-short pulses in linear and nonlinear media'',
Nonlinearity {\bf 18} (2005), 1351--1374.

\bibitem{Schafer} Y. Chung and T.Sch\"{a}fer,
``Stabilization of ultra-short pulses in cubic nonlinear media", Phys. Lett. A {\bf 361} (2007),
63--69.

\bibitem{CL09} M. Colin and D. Lannes, ``Short pulses approximations
in dispersive media'', SIAM J. Math. Anal. \textbf{41} (2009), 708--732.

\bibitem{CGL05} T. Colin, G. Gallice, and K. Laurioux, ``Intermediate models in nonlinear optics",
SIAM J. Math. Anal. {\bf 36} (2005), 1664–-1688.

\bibitem{GS05} M.D. Groves and G. Schneider, ``Modulating pulse solutions for quasilinear wave equations",
J. Diff. Eqs. {\bf 219} (2005), 221–-258.

\bibitem{Kal} L.A. Kalyakin, ``Asymptotic decay of a one-dimensional wave packet in a nonlinear
dispersive medium", Math USSR - Sb. {\bf 60} (1988), 457--483.

\bibitem{Kato} T. Kato, ``The Cauchy problem for quasi-linear symmetric hyperbolic systems",
Arch. Rat. Mech. Anal. {\bf 58} (1975), 181-–205.

\bibitem{Sch} P. Kirrmann, G. Schneider, and A. Mielke, ``The validity of modulation equations for extended
systems with cubic nonlinearities", Proc. Roy. Soc. Edinburgh Sect. A {\bf 122} (1992),
85-–91.

\bibitem{LPS09} Y. Liu, D. Pelinovsky, and A. Sakovich, ``Wave breaking in the
short-pulse equation", Dynamics of PDE {\bf 6} (2009), 291--310.

\bibitem{Matsuno1} Y. Matsuno, ``Multiloop soliton and multibreather solutions
of the short pulse model equation'', J. Phys. Soc. Japan {\bf 76} (2007), 084003, 6 pp.

\bibitem{Matsuno} Y. Matsuno, ``Periodic solutions of the short pulse
model equation'', J. Math. Phys. {\bf 49} (2008), 073508, 18 pp.

\bibitem{PS10} D. Pelinovsky, A. Sakovich, ``Global well-posedness of the short-pulse and
    sine--Gordon equations in energy space'', Comm. PDE {\bf 35} (2010), 613--629.

\bibitem{SS05} A. Sakovich and S. Sakovich, ``The short pulse
equation is integrable'', J. Phys. Soc. Japan \textbf{74} (2005), 239--241.

\bibitem{SS06} A. Sakovich and S. Sakovich, ``Solitary wave
solutions of the short pulse equation'', J. Phys. A: Math. Gen.
{\bf 39} (2006), L361--L367.

\bibitem{SW04}
T. Sch\"afer and C. E. Wayne, ``Propagation of
ultra-short optical pulses in cubic nonlinear media'',
Physica D, {\bf 196} (2004), 90--105.

\bibitem{SSK10} A. Stefanov, Y. Shen, and P.G. Kevrekidis, ``Well-posedness and small
data scattering for the generalized Ostrovsky equation", J. Diff. Eqs. {\bf 249} (2010), 2600--2617.

\bibitem{Kevrekidis} N.L. Tsitsas, T.R. Horikis, Y. Shen, P.G. Kevrekidis, N. Whitaker, and D.J.
Frantzeskakis, ``Short pulse equations and localized structures in frequency band gaps of nonlinear metamaterials",
Physics Letters A {\bf 374} (2010), 1384--1388.

\bibitem{Kutz} M.O. Williams, E. Shlizerman, and J.N. Kutz, ``The multi-pulsing transition in mode-locked lasers:
a low-dimensional approach using waveguide arrays", J. Opt. Soc. Am. B {\bf 27} (2010), 2471--2481.

\bibitem{Yin} Z. Yin, ``On the Cauchy problem for an integrable equation with peakon solutions",
Illinois J. Math. {\bf 47} (2003), 649--666.

\bibitem{Y1} A. Yoshikawa, ``Solutions containing a large parameter of a quasi-linear hyperbolic system of
equations and their nonlinear geometric optics approximation", Trans. Amer. Math. Soc. {\bf 340} (1993), 103--126.

\bibitem{Y2} A. Yoshikawa, ``Asymptotic expansions of the solutions to a class of quasi-linear hyperbolic
initial-value problems", J. Math. Soc. Japan {\bf 47} (1995), 227--252.



\end{thebibliography}
\end{document}